\documentclass{article}[12pt]
\usepackage{graphicx} 
\usepackage[english]{babel}
\usepackage[utf8]{inputenc}
\usepackage[T1]{fontenc}
\usepackage{lmodern}     
\usepackage{amsmath,amssymb} 
\usepackage{amsthm}
\usepackage{extarrows}
\usepackage{mathabx}
\usepackage{tabularx}    
\usepackage{multicol}    
\usepackage{multirow}    
\usepackage{longtable}   
\usepackage{booktabs}    
\usepackage{tikz}                
\usepackage{graphicx,graphics}   
\usepackage{hyperref}

\newcommand{\mc}{\ensuremath{\mathcal M}}
\newcommand{\ac}{\ensuremath{\mathcal A}}
\newcommand{\dbr}{\ensuremath{\}\!\}}}
\newcommand{\dbl}{\ensuremath{\{\!\{}}

\addto\captionsenglish{}

\theoremstyle{definition}
\newtheorem{definition}{Definition}

\theoremstyle{plain}
\newtheorem{theorem}{Theorem}
\newtheorem{proposition}{Proposition}

\newtheorem{corollary}{Corollary}

\newtheorem{remark}{Remark}

\title{Double Poisson brackets on low dimensional algebras}
\author{G.I. Sharygin \\ A. Hernandez Rodriguez}
\date{December 2024-January 2025}

\begin{document}

\maketitle
\begin{center}
{\textbf{Abstract}}
\end{center}

In this paper, we discuss the structure of double Poisson brackets in the sense of M. Van den Bergh on finite-dimensional algebras. In particular (using Hochschild homology and some results from \cite{vdb}) we prove that all possible double Poisson brackets on matrix algebras are given by inner derivations of appropriate type. As a corollary of this result, we see that all possible double Poisson brackets in any finite-dimensional semisimple algebra are also given by inner derivations. We further give a description of all double Poisson brackets on a small non-semisimple algebra (namely, the algebra of $2\times 2$ upper triangular matrices); in addition, we discuss Poisson structures induced from the double Poisson brackets in its representation spaces of dimensions two and three. Further following the development of this formalism, suggested by S. Arthamonov (see \cite{sart}), we compute all modified double Poisson brackets on this algebra, and the Poisson structures induced by these. 

\tableofcontents

\section{Introduction}
A classical approach to the study of integrable systems, based upon the Hamiltonian formalism, is to consider smooth manifold $M$ that substitutes the phase space, its ring of smooth functions $\mathcal{C}^\infty(M)$ that will represent the observables, and a Poisson bracket on $M$ which will induce the dynamics of the observables. Recall now that Poisson bracket is a (bi)linear operator 
\[
\{ , \}: \mathcal{C}^\infty(M) \otimes \mathcal{C}^\infty(M)\to \mathcal{C}^\infty(M), \]
with the following properties:
\begin{enumerate}
            \item Skew symmetry: $\{ f,g \}=-\{ g,f\}$
            \item Leibniz rule: $\{ fg,h\} =f\{ g,h\}+\{ f,h\} g$
            \item Jacobi identity: $\{\{f,g\},h\}+\{\{g,h\},f\}+\{\{h,f\},g\}=0$.
        \end{enumerate}
In effect from algebraic point of view, one can dispense with the manifold and still have a meaningful theory: the notion of Poisson structure can be generalized to a bilinear operation on an associative (not necessarily commutative) algebra $\mathcal{A}$ over a field $\mathbb{K}$ of characteristic zero: we shall say that $\{,\}$ is a \textit{Poisson bracket on \ac} if $\{,\}$ is a $\mathbb{K}$-(bi)linear operator
\[
\{,\}:\mathcal{A}\otimes \mathcal{A} \to \mathcal{A},
\]
satisfying the same properties of skew symmetry, Leibniz rule and Jacobi identity as above. Observe that the order of factors in the Leibniz rule is important if \ac\ is noncommutative; we will assume that they are fixed as it is written in the formula above.

\begin{remark}\rm
 When we fix one of the arguments in a Poisson bracket, we obtain a derivation of the algebra \ac; now recall that for semi-simple algebras all derivations are inner, which means that $\{a,b\}=[f(a),b]=[a,g(b)]$, where $[,]$ is the commutator and $f,g$ are $\mathbb{K}$-linear functions. From this and from the Jacobi identity one can deduce that in this case the only type of Poisson bracket there exists on \ac\ has the form $\{a,b\}=\lambda[a,b]$, where $[,]$ denotes the commutator of elements in the algebra and $\lambda\in\mathbb K$ is a constant. On the other hand, Poisson structures on free commutative and noncommutative algebras are quite rich and are subject of active study.
\end{remark}
From the remark, we can notice that interesting examples of Poisson structures for finite-dimensional noncommutative algebras can be rather limited; the purpose to study such structures arises from Noncommutative Geometry, where they can be used to develop the theory of ``integrable systems'', similarly to the usual (commutative) case. 

On the other hand, with every finitely-presented (in particular, every finite-dimensional) noncommutative algebra \ac\ one can associate an infinite series of \textit{representation spaces}:
\[
\mathrm{Rep}_n(\ac)=\{f:\ac\to\mathrm{Mat}_n(\mathbb K)\mid f(ab)=f(a)f(b)\}.
\]
These sets bear the structure of affine algebraic varieties: in effect, if \ac\ is generated by a finite set $a_1,\dots,a_N$ verifying the equalities $r_1,\dots,r_K$ (where every $r_i=r_i(a_1,\dots,a_N)$ is a noncommutative polynomial), then every homomorphism $f:\ac\to\mathrm{Mat}_n(\mathbb K)$ is uniquely determined by a finite set of matrices $f(a_1),\dots,f(a_N)\in\mathrm{Mat}_n(\mathbb K)$, which can be chosen arbitrarily as long as the equalities $r_i(f(a_1),\dots,f(a_N))=0,\,i=1,\dots,K$ hold. The group $GL_n(\mathbb K)$ acts on these spaces by conjugations and one can consider the quotient space 
\[
\widetilde{\mathrm{Rep}}_n(\ac)= \mathrm{Rep}_n(\ac)/Ad_{GL_n(\mathbb K)}.
\]
The representation spaces of certain algebras, such as the path algebras of quivers, play an important role in the theory of integrable systems serving as the phase spaces of important integrable systems; in particular they are equipped with Poisson structures; these Poisson structures have been extensively studied by various mathematicians in the last 30 years. 


%
It turns out that these structures arise from certain constructions related to the initial algebra \ac. Namely in 2007, Michel Van den Bergh \cite{vdb} introduced the concept of double Poisson brackets on an (not necessarily commutative) associative algebra; as is explained in the introduction to the cited paper, one of the incentives of this was to provide a suitable framework to explain the existence of Poisson structures on representation spaces of an algebra, with the ultimate purpose of studying dynamical systems on these spaces and relating their properties to the properties of the corresponding algebras. In the aforementioned article, M. Van den Bergh discussed other analogous structures and concepts related to a (possibly noncommutative) algebra, for instance: polyvector fields, Schouten-Nijenhuis brackets, Hamiltonians and others. He used this formalism to describe the canonical Poisson structures on representation spaces of quiver algebras. Later A.~Odeskii, V.~Rubtsov, and V.~Sokolov, classified and described double Poisson brackets on free associative algebras, and gave examples of noncommutative Hamiltonian integrable systems in terms of this formalism \cite{ors}. More recently B. Wang, and S. Li, used it to obtain a Lax pair of a discrete noncommutative dynamical system \cite{WL}.

There exist many generalizations and modifications of double Poisson brackets. For instance, in  \cite{sart} S. Arthamonov introduced a non-skew-symmetric generalization of a double Poisson bracket, the Modified Double Poisson bracket. Determined by the map into ``universal trace space'' (the quotient $\ac/ [ \ac,\ac]$) it also allows to induce a Poisson structure on moduli spaces of the space of representations of an algebra. Another approach was suggested by W. Crawley-Boevey, \cite{CBW}. All these constructions give insight into the way integrable systems on representation spaces behave.

In this paper, we attempt to study double Poisson structures (and some of their generalizations) on certain finite-dimensional algebras. We begin with the classification of double brackets on matrix algebras (and more generally any finite-dimensional semisimple algebras); as one can naturally expect, these structures are all given by commutators (up to a scalar factor). In particular, they induce trivial Poisson structures on representation spaces. Thus the same is true for all finite-dimensional semisimple algebras over real or complex numbers, since they are equal to direct sums of matrix algebras over real or complex fields, or quaternions by Wedderburn-Artin theorem. After this, we discuss the simplest non-semisimple algebra: the algebra of upper-triangular matrices over complex (or real) numbers. In this case there is essentially only one example of double Poisson structures, which induces nontrivial (but rather simple) Poisson structures on representation spaces; we describe these structures for representations in degrees $2$ and $3$. Finally, we describe Arthamonov's generalized double Poisson structures on the same algebra; we show that it looks similar to Van den Bergh's one, which we discussed earlier.

\section{Double Poisson brackets}

\subsection{Definitions and basic properties}
The following definition and propositions are taken from the paper \cite{vdb}:
\begin{definition}{\cite{vdb}}
Let \ac\ be an associative algebra over $\mathbb K$. A double Poisson bracket on \ac\ is a $\mathbb{K}$ (bi)linear map $\dbl-,-\dbr : \ac \otimes \ac \to \ac \otimes \ac $, such that the next conditions hold:
    \begin{itemize}
        \item This map is anti-symmetric in the following sense
        \[ \dbl a,b \dbr = -\dbl b,a \dbr^\circ,\]
        where $(a\otimes b)^\circ := b\otimes a$.
        \item This map is a derivation in its second argument with respect to the outer (bi)module
        structure on $\ac\otimes \ac$:
        \[\dbl a,bc \dbr=(b\otimes 1)\dbl a,c \dbr + \dbl a,b\dbr(1\otimes c).\]
        \item The \textit{double Jacobi identity} holds
        \[ \dbl a,\dbl b,c\dbr\dbr_L + \tau_{(123)} \dbl b,\dbl c,a\dbr\dbr_L + \tau_{(132)} \dbl c,\dbl a,b\dbr\dbr_L =0.\]
Here we put 
\[
\dbl a,b\otimes c\dbr_L =\dbl a,b\dbr\otimes c,
\]
and $\tau_{(123)}$ is the cyclic permutation of tensor legs: $\tau_{(123)}(a\otimes b \otimes c)=c\otimes a\otimes b$. 
    \end{itemize}
A $\mathbb{K}$-linear map that satisfies the first two properties is called just \textit{double bracket}. An algebra, equipped with the double Poisson bracket is called \textit{double Poisson algebra}.
\end{definition}
One of the important properties of double Poisson algebras is their relation with usual Poisson structures. In particular, the following proposition gives a construction, relating a double Poisson bracket on an algebra \ac\ with a variant of Poisson bracket on this algebra; this Poisson bracket will take values in \ac, but it will be anti-symmetric only after descending to the ``universal trace space'' of \ac, i.e. the \textit{quotient space} $\ac_\flat=\ac/[\ac,\ac]$ (here $[\ac,\ac]$ denotes the space of commutators in \ac):
\begin{proposition}{\cite{vdb}}
Assume that $(\ac,\dbl -,-\dbr)$ is a double Poisson algebra. Let 
\begin{equation}\label{eq:poisstr1}
\{-,-\}:\ac\otimes \ac \to \ac:a\otimes b \mapsto \dbl  a,b \dbr'\dbl a,b\dbr'',
\end{equation}
where $\dbl a,b\dbr=\dbl  a,b \dbr'\otimes\dbl a,b\dbr''$. Then the following statements are true:
    \begin{itemize}
        \item $\{-,-\}$ is a derivation in its second argument.
        \item $\{-,-\}$ is anti-symmetric modulo commutators, i.e. the composition of the bracket $\{-,-\}$ with the natural projection to $\ac_\flat$ is anti-symmetric.
        \item $\{-,-\}$ satisfies the following version of the Jacobi identity
        \[
        \{a,\{b,c\}\}=\{\{a,b\},c\}+\{b,\{a,c\}\}.
        \]
        \item Composition of the section map $\ac_\flat=\ac/[\ac,\ac]\to\ac$, bracket $\{-,-\}$ and projection $\ac\to\ac_\flat$ induces a Lie algebra structure on $\ac_\flat$ (in particular, it is well-defined as a map $\ac_\flat\otimes\ac_\flat\to\ac_\flat$).
    \end{itemize}
Remark that for commutative algebras $\ac_\flat=\ac/[\ac,\ac]=\ac$, so in the commutative case we get a Poisson structure on \ac\ in the usual sense.
\end{proposition}
In later research (see \cite{sart}) this property (i.e. the condition that the formula \eqref{eq:poisstr1} induces a Lie algebra structure on $\ac_\flat$) was used as the basic property in the definition of \textit{modified double Poisson structure} on \ac, see definition \ref{df:modifieddp}.



\subsection{The Poisson structure on $\mathrm{Rep}_n(\ac)$}
Recall that the $n$-dimensional $\mathbb K$-linear \textit{representation space} of an (finitely-pre\-sen\-ted) associative algebra \ac, $\mathrm{Rep}_n(\ac)$ is the set of all $\mathbb K$-linear homomorphisms $\ac\to\mathrm{Mat}_n(\mathbb K)$. Since every such homomorphism is determined by its values on generators of \ac, this space can be given the structure of affine variety inside $\mathbb K^{n^2N}$, where $N$ is the number of generators of \ac: the relations then correspond to the generating equations.

In order to induce a Poisson structure in the representation space $\mathrm{Rep}_n(\ac)$ of $\ac$ for a given $n\in\mathbb{N}$, it is convenient to consider \textit{the coordinate ring} or \textit{the ring of regular functions} on (the affine variety) $\mathrm{Rep}_n(\ac)$, denoted $\mathcal{O}(\mathrm{Rep}_n(\ac))$. Then every element $a\in\ac$ determines a collection of functions $a_{ij}\in\mathcal{O}(\mathrm{Rep}_n(\ac)),\,i,j=1,\dots,n$. Namely, for every $a\in \ac$ we define the element $a_{ij}\in \mathcal O(\mathrm{Rep}_n(\ac))$ as
\[
a_{ij}(x)=(x(a))_{ij},~ x\in \mathrm{Rep}_n(\ac),~i,j=1,\dots,n.
\]

Suppose now that \ac\ is equipped with a double Poisson bracket. As Van den Bergh explains in the cited paper, this allows one to induce a Poisson structure in the ring $\mathcal{O}(\mathrm{Rep}_n(\ac))$. Indeed, consider the formula:
\[\{a_{ij},b_{pq}\}=\dbl  a,b \dbr_{pj}'\cdot \dbl  a,b \dbr_{iq}'',~a,b\in A,~i,j,p,q=1,\dots,n.\]

\subsection{Some relevant results}
In usual ``commutative'' geometry Poisson brackets are often described by bivector fields. Looking for an analogous structure in noncommutative algebras one comes up with the concept of (double) polyvector fields on $\ac$.
Recall that derivation of an algebra \ac\ with values in a bimodule \mc\ is a linear map $X:\ac\to\mc$ satisfying the Leibniz equation:
\[
X(ab)=X(a)b+aX(b).
\]
Then we have the following definition:
\begin{definition}
Let \ac\ be an associative algebra. Double vector fields on $\ac$ are given by the set of derivations $D_\ac:=Der(\ac,\ac\otimes \ac)$, where we use ``external'' bimodule structure on $\ac\otimes\ac$ in the definition of derivations.
\end{definition}
We can now define (double) polyvector fields on an associative algebra \ac\ (remark that the existence of the second \ac-bimodule structure on $\ac\otimes\ac$ is crucial for this construction):
\begin{definition}
The algebra of (double) polyvector fields $D^\bullet_\ac$ on $\ac$ is the tensor algebra $T_\ac D_\ac$ of $D_\ac$ over \ac\ where we give the structure of an $\ac$-(bi)module to $D_\ac$ using the inner (bi)module structure on $\ac\otimes \ac$.
\end{definition}
Recall that a derivation $\xi:\ac\to\mc$ (for an \ac-bimodule \mc) is called ``inner'' if it is given by 
\[
\xi(a)=am-ma,\ a\in\ac
\]
for some $m\in\mc$. Similarly, we will call an element $\xi\in D^\bullet_\ac$ ``inner'', if it is given by a tensor product of inner derivatives.

In the paper cited, Van den Bergh shows that every ``polyvector field'' in $D^n_\ac$ gives rise to a multi-linear ``double''-bracket on \ac, i.e. a polylinear map
\[
\dbl -,\dots,-\dbr:\ac^{\otimes n}\to\ac\otimes\ac.
\]
Since we are only interested in double Poisson brackets, i.e. in case $n=2$, let us give definition in this case only:
\begin{proposition}{\cite{vdb}}\label{prop4.1.1.}
    There is a well defined linear map
    \[\mu:D_\ac^2\to \{\mbox{bilinear double brackets on }\ac\}:Q\mapsto \dbl -,-\dbr_Q\]
    which on $Q=\delta_1\otimes\delta_2$ is given by
    \[\dbl-,-\dbr_Q=\sum_{i=0}^{1}(-1)^{i}\tau_{(12)}^i\circ\dbl-,-\dbr^\sim_Q\circ\tau_{(12)}^{-i} \]
    where
    \[\dbl a_1,a_2 \dbr^\sim_{Q}=\delta_2(a_2)'\delta_1(a_1)''\otimes\delta_1(a_1)'\delta_2(a_2)''\]
    and $\tau_{(12)}$ is the permutation of the tensor legs.
\end{proposition}
In the cited paper, Van den Bergh introduced the analog of the Schouten-Nijenhuis bracket on $D_\ac$ and showed that \textit{every double Poisson bracket is induced by a double bivector field, whose square with respect to the Schouten bracket vanishes modulo commutators}.

The importance of the ``double polyvectors'' is in great measure due to the role they play in the theory of double brackets. Namely, assume that $\ac$ is a finitely generated $\Bbbk$-algebra. We recall that (according to Van den Bergh) an associative algebra \ac\ is called \textit{smooth} (see \cite{vdb}) if the \ac-bimodule $\Omega(\ac)=\ker(\ac\otimes\ac\to\ac)$ is projective (as bimodule). Then 
\begin{proposition}{\cite{vdb}}\label{prop4.1.2.}
    If $\ac$ is smooth then $\mu$ is an isomorphism.
\end{proposition}
Thus, in the smooth case, the classification of all brackets amounts to the classification of double polyvectors. Below we will make use of this fact.

\subsection{Inner double brackets}
Using the notation from the previous section, let $Q=\delta_1\otimes\delta_2$ be an inner bivector in $D^2_\ac$, i.e. 
\[
\delta_1(a)=[U\otimes V,a],\,\delta_2(a)=[X\otimes Y,a],
\]
then
\[
\dbl a,b \dbr^\sim_{Q}=[[XV\otimes UY,a]_{in},b]_{out},
\]
where $[,]_{in},\,[,]_{out}$ denote the commutators with respect to inner/outer bimodule structure, i.e.
\[
\dbl a,b \dbr^\sim_{Q}=bXVa\otimes UY-bXV\otimes aUY-XVa\otimes UYb+XV\otimes aUYb.
\]
Eventually we see that $\dbl a,b \dbr_{Q}$ in this case is given by
\[
\dbl a,b \dbr_{Q}=\dbl a,b \dbr^\sim_{Q}-\tau_{12}(\dbl b,a \dbr^\sim_{Q})=[[XV\otimes UY-UY\otimes XV,a]_{in},b]_{out}.
\]
It is clear that if \ac\ is unital (which we will mutely assume from this moment), the elements $XV\otimes UY-UY\otimes XV$ span the whole subspace of skew-symmetric tensors in $\ac\otimes \ac$. From now on we will abbreviate this construction to
\[
\dbl x,y\dbr_r=[[r,x]_{in},y]_{out},\ r\in\Lambda^2\ac.
\]
The condition that double brackets satisfy the Jacobi identity, can be expressed in terms of $r$ as follows: let
\[
J(r)=r_{13}\times r_{12}+r_{23}\times r_{13}-r_{12}\times r_{23},
\]
where $r_{ij}$ denotes the natural inclusion of $r$ into $\ac^{\otimes 3}$ as the $i$-th and the $j$-th tensor legs. Then the double Jacobi identity for $\dbl,\dbr_r$ holds iff
\[
[[[J(r),x]_1,y]_2,z]_3=0\ \mbox{for all}\ x,y,z\in\ac.
\]
Here $[J,x]_k$ for $J\in\ac^{\otimes 3},\,x\in\ac$ denotes the commutator of $J$ and $x$ with respect to one of the three natural bimodule structures in $\ac^{\otimes 3}$:
\[ 
\begin{aligned}
{}[a\otimes b \otimes c,x]_1&=a\otimes xb\otimes c - ax\otimes b \otimes c,\\
[a\otimes b \otimes c,y]_2&=a\otimes b\otimes yc - a\otimes b y\otimes c,\\
[a\otimes b \otimes c,z]_3&=za\otimes b\otimes c - a\otimes b \otimes c z.
\end{aligned}
\]
One can prove this by a straightforward (albeit quite tedious) computation. In particular, \textit{the double Jacobi identity for an inner double bracket $\dbl,\dbr_r$ holds if the element $r\in\Lambda^2\ac$ satisfies the \textbf{classical associative Yang-Baxter equation}}
\begin{equation}\label{eq:ybea}
r_{13}\times r_{12}+r_{23}\times r_{13}-r_{12}\times r_{23}=0.
\end{equation}
We will call such double Poisson structures \textit{inner}. In other words, inner double Poisson structures are the structures, given by commutators with the elements $r\in\ac\wedge\ac$. 

Finally, assume that the algebra \ac\ is smooth (see proposition \ref{prop4.1.2.} and discussion before it), in this case all double brackets (not necessarilly satisfying the Jacobi identity) are induced by ``double bivectors''. In particular, if all ``double bivectors'' are inner, then so are all double (Poisson) brackets.

\begin{remark}
Let us make one more observation about the properties of inner double brackets: \textit{if a double Poisson bracket is inner, then functions $tr x,\,x\in\ac$ in the corresponding Poisson structure on representation space are central}. Here $tr x=\sum_i x_{ii},\,x\in\ac$ in the notation we introduced above. This can be shown by a direct computation: if $r=A\wedge B\in\Lambda^2\ac$, then
\[
\begin{split}
\sum_i\{x_{ii},y_{pq}\}&=\dbl  x,y \dbr_{pi}'\cdot \dbl  x,y \dbr_{iq}''\\
                        &=(Ax)_{pi}(By)_{iq} - (yAx)_{pi}B_{iq} - A_{pi}(xBy)_{iq} + (yA)_{pi}(xB)_{iq}\\
                        &\quad -(Bx)_{pi}(Ay)_{iq} + (yBx)_{pi}A_{iq} + B_{pi}(xAy)_{iq} - (yB)_{pi}(xA)_{iq}\\
                        &=(AxBy)_{pq}-(yAxB)_{pq}-(AxBy)_{pq}+(yAxB)_{pq}\\
                        &\quad-(BxAy)_{pq}+(yBxA)_{pq}+(BxAy)_{pq}-(yBxA)_{pq}=0.
\end{split}
\]
\end{remark}

\subsection{Modified Double Poisson brackets}
There exist few generalizations of double Poisson brackets in the literature. All these structures are motivated by the desire to have an algebraic structure on \ac, which can induce Poisson structures on representation spaces (usual or reduced) of the algebra. Here we give definitions of two such structures.

First, Crawley-Boevey introduced the concept of $H_0$-Poisson structure (see \cite{CBW}), as a structure on the universal trace space $\ac_\flat=\ac/[\ac,\ac]$, verifying certain conditions. This is an algebraic structure, that also induces a Poisson bracket on the reduced representation space $\widetilde{\mathrm{Rep}}_n(\ac)= \mathrm{Rep}_n(\ac)/Ad_{GL_n(\mathbb K)}$. The term $H_0$-structure is used because the universal trace space, the quotient $\ac_\flat=\ac/[\ac,\ac]$, coincides with the $0$-th Hochschild homology group of \ac, $H_0(\ac)$.

\begin{definition}
    A linear map $\{ -,- \}_{H_0}\colon \ac_\flat\otimes \ac_\flat\to \ac_\flat$ is called a \textit{$H_0-$Poisson structure on $\ac$} if
    \begin{itemize}
    \item The bracket $\{ -,- \}_{H_0}$ is a Lie bracket on $\ac_\flat$
    \item For every $a\in \ac_\flat$ the linear map $\{a,- \}_{H_0}\colon \ac_\flat\to \ac_\flat$ is induced by a derivation on $\ac$.
    \end{itemize}
\end{definition}
This construction was later further modified by Arthamonov, see \cite{sart}, who coined the notion of \textit{modified} double Poisson structure; roughly speaking it is a version of Van den Bergh's double Poisson structure, where certain conditions are relaxed:
\begin{definition}\label{df:modifieddp}
    A linear map $\dbl -,-\dbr\colon \ac \otimes \ac \to \ac \otimes \ac$ is called a \textit{modified double Poisson bracket} if it satisfies the following properties for any $a,b,c\in \ac$,
    \begin{itemize}
        \item Leibniz rules, \[\dbl a,bc \dbr=(b\otimes 1)\dbl a,c \dbr + \dbl a,b\dbr(1\otimes c),\]
        \[\dbl ab,c \dbr = (1\otimes a) \dbl b,c \dbr + \dbl a,c \dbr (b\otimes 1) .\]
        \item $H_0-$skewsymmetry, \[ \{a,b\}+\{b,a\}\in [\ac,\ac], \]
        where $\{-,-\}=m\circ \dbl -,-\dbr$, and $m$ is just the product of the elements in $\ac$; in other words the operation $\{,\}$ is skew-symmetric as a map to $\ac_\flat$.
        \item Jacobi identity, \[ \{a,\{b,c\}\}-\{ b,\{a,c\}\}-\{\{a,b\},c\} =0. \] 
    \end{itemize}
\end{definition}
These two structures are closely related. In fact, one can show (see \cite{sart}) that for any modified double Poisson structure $\dbl,\dbr$ on \ac\ the formula
\[
\{x,y\}_{H_0}=(m\circ\dbl\bar x,\bar y\dbr)_\flat
\]
where $\bar x,\bar y$ are some representatives in \ac\ of the elements $x,y\in\ac_\flat$ and for any $a\in\ac$ we denote by $a_\flat$ its natural projection to the universal trace space, determines an $H_0$-Poisson structure on \ac. Informally one can say that Arthamonov's modified double Poisson bracket is given by choosing representatives of an $H_0$-bracket of Crawley-Boevey.

The importance of these two structures is justified by the following
\begin{theorem}[\cite{sart}]
Any $H_0-$Poisson structure on $\ac$ induces a unique Poisson bracket on the quotient $\widetilde{\mathrm{Rep}}_n(\ac)$. In particular, any modified double Poisson bracket on $\ac$ induces a unique Poisson bracket on the quotient $\widetilde{\mathrm{Rep}}_n(\ac)$.
\end{theorem}
Below we will give an example of modified double brackets in this sense of Arthamonov.

\section{Double Poisson brackets on matrix algebras}
The main purpose of this paper is to give examples of double Poisson structures on finite-dimensional algebras. We begin with matrix algebras.

\subsection{Equations for a double Poisson bracket}
In general, double Poisson brackets on a finite-dimensional algebra \ac\ can be described as solutions to a system of linear and quadratic equations on their coefficients with respect to a fixed basis in \ac. For instance, in the case of matrix algebras $\mathrm{Mat}_n(\mathbb{K}),\, n\in\mathbb N$ we 
can consider the basis of matrix units $E_{ij},\,i,j=1,\dots,n$; recall that $E_{ij}$ is the matrix with entries 
\[
(E_{ij})_{\alpha\beta}=\delta_{i\alpha}\delta_{j\beta}.
\]
Then any double Poisson bracket can be described in terms of this basis in terms of the corresponding coefficients:
\[
\dbl E_{ij},E_{kl}\dbr=\Phi_{ijkl}^{\alpha\beta\gamma\delta}E_{\alpha\beta}\otimes E_{\gamma\delta}.
\]
Now the problem is to determine such coefficients $\Phi_{ijkl}^{\alpha\beta\gamma\delta}$ that the skew-symmetricity, Leibniz rules and the double Jacobi identity would hold.

First, consider the anti-symmetricity condition. It amounts to the following equations:
\[
\begin{split}
\Phi_{ijkl}^{abcd}&E_{ab}\otimes E_{cd}=\dbl E_{ij},E_{kl}\dbr=
- \dbl E_{kl},E_{ij}\dbr^\circ= -(\Phi_{klij}^{\alpha\beta\gamma\delta}E_{\alpha\beta}\otimes E_{\gamma\delta})^\circ\\
 & =-\Phi_{klij}^{\alpha\beta\gamma\delta} E_{\gamma\delta}\otimes E_{\alpha\beta}=-\Phi_{klij}^{cdab} E_{ab}\otimes E_{cd}
\end{split}
\]
Thus the skew-symmetricity of the bracket is equivalent to the equation
\[
\Phi_{ijkl}^{abcd}=-\Phi_{klij}^{cdab}
\]
for all $i,j,k,l,a,b,c$ and $d$. To obtain the Leibniz rule with respect to the second argument consider the structure constants of the matrix algebra
\[
E_{ij}E_{kl}=\delta_{jk}E_{il}=\sum_{p,q=1}^n\delta_{jk}\delta_{pi}\delta_{ql} E_{pq}=C_{ijkl}^{pq}E_{pq}.
\]
Now expanding on both sides the equation
\[
\dbl E_{ab},C_{ijkl}^{pq}E_{pq}\dbr =(E_{ij}\otimes Id_n) \dbl E_{ab},E_{kl}\dbr + \dbl E_{ab},E_{ij}\dbr (Id_n\otimes E_{kl}),
\]
we get
\[
\delta_{\alpha i}\Phi_{abkl}^{j\beta\gamma\delta}+\delta_{\delta l}\Phi_{abij}^{\alpha\beta\gamma k} - \delta_{jk} \Phi_{abil}^{\alpha\beta\gamma\delta}=0.
\]
The double Jacobi identity has the form
\[
\begin{split} 
\dbl E_{ij},\dbl E_{kl} ,E_{mn} \dbr\dbr_L &+ \tau_{(123)} \dbl E_{mn},\dbl E_{ij} ,E_{kl} \dbr\dbr_L \\
                                                                &\qquad\qquad+\tau_{(132)} \dbl E_{kl},\dbl E_{mn} ,E_{ij} \dbr \dbr_L =0,
\end{split}
\]
and after some computations we see that it is satisfied if
\[
\Phi_{klmn}^{\alpha\beta\gamma\delta}\Phi_{ij\alpha\beta}^{wxyz} + \Phi_{mnij}^{\alpha\beta w x} \Phi_{kl\alpha \beta}^{yz\gamma\delta} + \Phi_{ijkl}^{\alpha\beta yz} \Phi_{mn\alpha\beta}^{\gamma\delta wx} =0.
\]
Thus in total, we have an overdetermined system of linear and quadratic equations on $n^8$ coefficients. Even in case $n=2$ that would give 256 variables and more than 1000 equations.

\subsection{Smoothness of the algebra of matrices}
As one can see, the system of equations that determine double Poisson structures is quite involved. So, to avoid the brute-force computations, we try to reduce the difficulty by proving that all double brackets in this situation are inner, thus the double Poisson brackets are classified by the solutions of classical associative Yang-Baxter equation \eqref{eq:ybea}. Recall that according to Van den Bergh all the double multibrackets on a smooth algebra are given by ``double polyvector fields'' see proposition \ref{prop4.1.2.}. So, we begin by showing that the matrix algebras are smooth.

 It is obvious that all matrix algebras are finitely generated, thus it remains to prove that the bimodule $\Omega(\mathrm{Mat}_n(\mathbb{K}))$ is projective over $\mathrm{Mat}_n(\mathbb{K})$. 
In other words, we must show that it is a direct summand of a free $\mathrm{Mat}_n(\mathbb{K})$-bimodule.

   
So we prove
\begin{proposition}
    Let $n\in\mathbb{N},\ \ac=\mathrm{Mat}_n(\mathbb{K})$. Then $\Omega(\ac)=\ker(m:\ac\otimes\ac\to\ac)$ is a projective $\ac\otimes \ac^\circ$-module (here $\ac^\circ$ denotes the opposite algebra).  
\end{proposition}
\begin{proof}
We can identify the algebra $\ac\otimes \ac^\circ$ with $\mathrm{Mat}_{n^2}(\mathbb{K})\cong \mathrm{Mat}_{n}(\mathbb{K})^{\otimes^2}$; the identification is $x\otimes y\mapsto x\otimes y^\top$. In this case the algebra $\ac$ as $\ac\otimes\ac^\circ$-module is identified with the (left) $\mathrm{Mat}_{n^2}(\mathbb{K})$-module $\mathbb K^{n^2}$, the standard representation of the matrix algebra. And this module is a direct summand in the free $\mathrm{Mat}_{n^2}(\mathbb{K})$-module $\mathrm{Mat}_{n^2}(\mathbb{K})$.

On the other hand, by definition of $\Omega(\ac)$, we have the exact sequence of bimodules
\[
0\to \Omega(\ac)\to \ac\otimes \ac\stackrel{m}{\to} \ac\to 0
\]
Since $\ac$ is a projective $\mathrm{Mat}_{n^2}(\mathbb{K})$-module, this sequence is split, and it follows that $\Omega(\ac)$ is direct summand of $\ac\otimes \ac$, which is a free $\ac\otimes \ac$-module.
\end{proof}

\subsection{Double Poisson brackets of the algebra of matrices are inner}
Now as a concequence of the previous observation we can use the result of the proposition \ref{prop4.1.2.}: all multibrackets on a matrix algebra are induced by double polyvectors. In order to show that all the polyvectors are inner, we will need one more definition:
\begin{definition}\cite{morita}
 We say that rings $R,S$ are Morita equivalent if the (left) module categories $_RMod,~_SMod$ are equivalent.
\end{definition}
In particular, it is well known that for any ring $R$ is Morita equivalent to the corresponding matrix algebras $\mathrm{Mat}_n(R)$ for every $n\in\mathbb{N}$.
An important property of Morita equivalent algebras is that \textit{their Hochschild (co)homology are equal in all dimensions}, see for instance the book by Loday, \cite{Loday}. Now, using all these observations we have the following result.
\begin{proposition}
For any $n\in\mathbb{N}$ all double Poisson brackets on $\mathrm{Mat}_n(\mathbb{K})$ are inner.
\end{proposition}
\begin{proof}
    Let $\ac=\mathrm{Mat}_n(\mathbb{K}),~\mathcal{B}=\mathbb{K}$, consider $D_{\ac}=Der_\mathcal{B}(\ac,\ac\otimes \ac),~D_\ac=T_\ac D_{\ac}$. As we have already proved $\ac$ is smooth as $\mathcal{B}$-algebra. It follows from propositions (\ref{prop4.1.1.},\ref{prop4.1.2.}) that there is an isomorphism between all double  Poisson brackets and $D_\ac^2$.

On the other hand, we know that Hochschild cohomology is Morita invariant. It is also well-known that for an associative algebra $C$ and a $C$-(bi)module $M$, the first  Hochschild cohomology space is given by the quotient of all $M$-valued derivations of $C$ modulo inner derivations:
\[
HH^1(C,M)=Der(\ac,M)/Inn(\ac,M)=Out(\ac,M)
\]
On the other hand, by Morita equivalence, 
\[
HH^1(\ac,\ac\otimes \ac)\cong HH^1(\mathcal{B},\ac\otimes\ac)=0,
\]
hence all derivations in $D_{\ac}$ are inner.
\end{proof}

In particular, for this case, it follows that double brackets on matrix algebras are inner, i.e. they are linear combinations of the following expressions for some $A,B\in \mathrm{Mat}_n(\mathbb{K})$:
\[
\begin{split}
\dbl x,y\dbr_{A\wedge B}&=[[A\otimes B-B\otimes A,x]_{in},y]_{out}\\
&= Ax\otimes By - yAx\otimes B - A\otimes xBy + yA\otimes xB\\
&\quad -Bx\otimes Ay + yBx\otimes A + B\otimes xAy - yB\otimes xA .
\end{split}
\]
As a simple corollary, we get the following 
\begin{corollary}
 All double Poisson brackets on all finite-dimensional semisimple algebras over algebraically closed field $\bar{\mathbb{K}}$ (in particular, on all finitely generated group algebras over $\mathbb C$) are inner.
\end{corollary}
\begin{proof}

Recall that all semisimple algebras over algebraically closed fields are equal to direct sums of group algebras, i.e. 
\[
\ac\cong \bigoplus_{i=1}^{N}\mathrm{Mat}_{n_i}(\bar{\mathbb{K}})
\]
for some $N,n_i\in\mathbb{N}$. Hence these algebras are smooth and are Morita equivalent to the algebra $\bar{\mathbb{K}}^N$ with coordinate-wise product. Hence their first Hochschild cohomology is given by direct sums of
\[
HH^1(\bar{\mathbb{K}},\bar{\mathbb{K}}\otimes\bar{\mathbb{K}})=0.
\]
\end{proof}

From the corollary we deduce that for finite dimensional algebras, interesting (non-inner) examples of double Poisson brackets can be found only for non-semisimple algebras. This is the motivation for the next
sections, where we give the description of double Poisson structures on a non-semisimple algebra.

\section{Example of a Double Poisson bracket on a non-semisimple algebra}
The simplest nontrivial example of a finite-dimensional non-semisimple algebra, is the algebra $\ac_2$ of real or complex upper triangular matrices of size $2\times 2$. It has $3$ linear generators 
\[
e_0=\begin{pmatrix}0 & 1\\ 0 & 0\end{pmatrix}, e_1=\begin{pmatrix}1 & 0\\ 0 & 0\end{pmatrix}, e_2=\begin{pmatrix}0 & 0\\ 0 & 1\end{pmatrix},
\]
which satisfy the following system of equations:
\begin{itemize}
    \item $e_1+e_2=1$,
    \item $e_1^2=e_1~,~~e_2^2=e_2~,~~e_1e_0=e_0e_2=e_0$,
    \item $e_0e_1=e_2e_0=e_0^2=0$.
\end{itemize}
This algebra is clearly isomorphic to the quiver algebra of the oriented graph
\[
\bullet\to\bullet
\] 
so that the idempotents $e_1,\,e_2$ correspond to the vertices and the element $e_0$ corresponds to the arrow. This algebra is non-semisimple, due to the existence of nilpotent ideal $\langle e_0\rangle$ in it.

\subsection{Computation for double Poisson bracket}
In order to describe the double brackets in $\ac_2$, we take the basis $e_0,e_1,e_2$ in $\ac_2$ and the corresponding basis $\ac_2\otimes\ac_2$. 
For our purposes it will be convenient to sometimes add the element $1= e_1+e_2$ to this basis; then the double bracket is determined by the formula
\[
\dbl e_i,e_j \dbr= C_{ij}^{\alpha\beta} e_\alpha\otimes e_\beta.
\]
\begin{remark}\label{remac2}
The advantage of introducing the unit $1$ in the list of basis elements is caused by the formula 
\[
\dbl 1,x\dbr=\dbl x,1\dbr=0\ \mbox{for all}\ x\in\ac_2,
\] 
which is the direct consequence of the Leibniz rule. Also we have relation
 \[
 C^{2j}e_2\otimes e_j=C^{2j} 1 \otimes e_j - C^{2j} e_1\otimes e_j.
 \] 
So we will consider the basis $\{1\otimes 1,e_0\otimes 1 ,1\otimes e_0,e_0\otimes e_0, 1\otimes e_1,e_1\otimes 1, e_0\otimes e_1,e_1\otimes e_0, e_1\otimes e_1\}$ of $\ac_2\otimes\ac_2$. 
For the same reason, it is enough to consider only the double brackets $\dbl e_0,e_0 \dbr,\dbl e_0,e_1 \dbr,\dbl e_1,e_1 \dbr$; in all the remaining cases we have
    \[
    \begin{split}
    \dbl e_2,e_2 \dbr&=\dbl 1-e_1,1-e_1 \dbr=\dbl e_1,e_1 \dbr,\\
    \dbl e_0,e_2 \dbr&=\dbl e_0,1-e_1 \dbr=-\dbl e_0,e_1 \dbr,\\
    \dbl e_1,e_2 \dbr&=\dbl e_1,1-e_1 \dbr=-\dbl e_1,e_1 \dbr. 
    \end{split}
    \]
\end{remark}
We are about to solve the equations on coefficients $C_{ij}^{\alpha\beta}$; as we remarked above, the general case is reduced to the situation when $i,j=0,1$. Then we have the following relations:
\begin{itemize}
\item Because of skew-symmetricity, we have $C_{ij}^{\alpha\beta}=-C_{ji}^{\beta\alpha}$, so $C_{ii}^{\alpha\alpha}=0$.
\item Since $1=e_1+e_2$, we have
\[
\dbl e_i ,1 \dbr=0=\dbl e_i,e_1 \dbr+\dbl e_i,e_2 \dbr\implies \dbl e_i,e_1 \dbr=-\dbl e_i ,e_2 \dbr,
\]
which means $C_{i1}^{\alpha\beta}=-C_{i2}^{\alpha\beta}$. In particular, due to the skew-symmetricity this gives
\[\begin{aligned}
C_{11}^{\alpha\alpha}=C_{12}^{\alpha\alpha}&=C_{21}^{\alpha\alpha}=C_{22}^{\alpha\alpha}=0,\\
C_{01}^{\alpha\beta}&=-C_{02}^{\alpha\beta}.
\end{aligned}
\]
\item Again by skew-symmetricity $C_{ii}^{\alpha\beta}=-C_{ii}^{\beta\alpha}$.

Here and below the index $\cdot$ is used to represent the element $1$ in the tensor product (see remark \ref{remac2}).
\end{itemize}
Now, let us solve the equations that guarantee Leibniz rule in the second argument (then the corresponding identity for the first argument will follow by skew-symmetricity):
\begin{itemize}
\item $\dbl e_i , e_1\dbr=\dbl e_i,e_1^2 \dbr=e_1\dbl e_i ,e_1 \dbr+\dbl e_i,e_1 \dbr e_1$.\\    
    {\textbf{i=1.}} In this case we get $C_{11}^{\cdot 0}=C_{11}^{01}=0$.

    {\textbf{i=0.}} For this case we conclude $C_{01}^{\cdot\cdot}=C_{01}^{\cdot 0}=0,$ $C_{01}^{0\cdot}=-C_{01}^{01}$, $C_{01}^{11}=-C_{01}^{\cdot 1}-C_{01}^{1\cdot}$.

    \item $\dbl e_i,e_0^2 \dbr=\dbl e_i,0 \dbr=0=e_0\dbl e_i ,e_0 \dbr+\dbl e_i,e_0 \dbr e_0$

    {\textbf{i=0.}} Here we have $C_{00}^{\cdot 1}=C_{00}^{01}=0$.

    {\textbf{i=1.}} For this case we get $C_{01}^{1\cdot}=0$, $C_{01}^{0\cdot}=-C_{01}^{10}$.

    \item $\dbl e_i,e_0 \dbr=\dbl e_i,e_1e_0 \dbr=e_1\dbl e_i,e_0 \dbr+\dbl e_i,e_1 \dbr e_0$

    {\textbf{i=1.}} In this case we get $C_{01}^{0\cdot}=-C_{11}^{\cdot 1}$.

    {\textbf{i=0.}} For this one, we get that $C_{01}^{\cdot 1}=C_{00}^{\cdot 0}$.
\item Finally, the equations
\[
\dbl e_0,0\dbr =0 = \dbl e_0,e_0e_1\dbr=e_0\dbl e_0,e_1\dbr+\dbl e_0,e_0\dbr e_1
\]
and
 \[
 \dbl e_1,0\dbr =0 = \dbl e_1,e_0e_1\dbr=e_0\dbl e_1,e_1\dbr+\dbl e_1,e_0\dbr e_1
 \]
already hold with the above conditions.
\end{itemize}

As far as the double Jacobi identity is concerned, we see that it holds automatically in some cases, namely:
\[
\dbl e_0,\dbl e_0,e_0\dbr\dbr_L + \tau_{(123)} \dbl e_0,\dbl e_0,e_0\dbr\dbr_L + \tau_{(132)} \dbl e_0,\dbl e_0,e_0\dbr\dbr_L=0,
\]
and
\[
\dbl e_1,\dbl e_1,e_1\dbr\dbr_L + \tau_{(123)} \dbl e_1,\dbl e_1,e_1\dbr\dbr_L + \tau_{(132)} \dbl e_1,\dbl e_1,e_1\dbr\dbr_L=0.
\]
It remains to check few remaining sets of the arguments: from the equations
\[
\dbl e_0,\dbl e_0,e_1\dbr\dbr_L + \tau_{(123)} \dbl e_0,\dbl e_1,e_0\dbr\dbr_L + \tau_{(132)} \dbl e_1,\dbl e_0,e_0\dbr\dbr_L,
\] and
\[
\dbl e_0,\dbl e_1,e_1\dbr\dbr_L + \tau_{(123)} \dbl e_1,\dbl e_1,e_0\dbr\dbr_L + \tau_{(132)} \dbl e_1,\dbl e_0,e_1\dbr\dbr_L=0,
\]
we get that the coefficients of the bracket should satisfy the condition ${C_{01}^{0\cdot}}^2=- C_{01}^{00}C_{01}^{\cdot 1}$.

Summing up we see that any double Poisson bracket on $\ac_2$ has the following form:
\begin{align*}
\dbl e_0,e_1 \dbr = \alpha(e_0 \otimes e_0) + \beta ( 1 \otimes e_1 &- e_1 \otimes e_1) + \gamma (e_0\otimes 1 - e_0 \otimes e_1 - e_1\otimes e_0),\\
\dbl e_0,e_0 \dbr &= \beta (1\otimes e_0 - e_0\otimes 1) ,\\
\dbl e_1,e_1 \dbr &= \gamma (e_1\otimes 1 - 1\otimes e_1), 
\end{align*}
with parameters $\alpha,\beta,\gamma \in \mathbb{K}$ satisfying the equation
\[
\gamma^2 +\alpha \beta = 0.
\]
All the other commutation relations follow from these three. In particular, we can take $\beta=\gamma=0$ and $\alpha\in\mathbb K$ arbitrary element, so that $\dbl e_0,e_1 \dbr_\alpha = \alpha(e_0 \otimes e_0)$, while all the other relations are trivial.
\begin{remark}
It is worth checking, if the double Poisson brackets on $\ac_2$ described above are inner, or not. For instance, let's take $r=e_0\wedge e_1$. In this case the expression $J(r)$ (see equation \eqref{eq:ybea}) vanishes, so the inner double bracket $\dbl,\dbr_{e_0\wedge e_1}$ satisfies the double Jacobi identity. By direct computations we get
\[
\begin{aligned}
\dbl e_0,e_0\dbr_{e_0\wedge e_1}&=0,\\
\dbl e_1,e_1\dbr_{e_0\wedge e_1}&=0,\\
\dbl e_0,e_1\dbr_{e_0\wedge e_1}&=e_0\otimes e_0.
\end{aligned}
\]
Also, if we take $r=1\wedge e_0$, then
\[
J(r)=-(e_0\otimes e_0\otimes 1+e_0\otimes 1\otimes e_0+1\otimes e_0\otimes e_0)\neq0,
\]
however, by direct computation we have
\[
[[[ (1\otimes e_0\otimes e_0 + e_0\otimes e_0\otimes 1 + e_0\otimes 1\otimes e_0), x]_1, y]_2, z]_3=0,
\] 
for any $x,y,z\in \ac_2$. Indeed, it is enough to take $x,y,z=e_0,e_1$, since this expression will clearly vanish if any of the arguments is $1$. In this case we have
\[
\begin{aligned}
\dbl e_0,e_0\dbr_{1\wedge e_0}&=0,\\
\dbl e_1,e_1\dbr_{1\wedge e_0}&=0,\\
\dbl e_0,e_1\dbr_{1\wedge e_0}&=-2e_0\otimes e_0.
\end{aligned}
\]
We see that the double Poisson bracket $\dbl,\dbr_\alpha$, described above, is inner.

More generally, if we take a generic element $a 1\wedge e_0+b 1\wedge e_1+ce_0\wedge e_1$ in $\Lambda^2\ac_2$ and check the associative Yang-Baxter equation for it, we get the following system of equations
\[
\begin{aligned}
ac&=a^2,& ab&=0, & ab&=bc, & b^2&=0. 
\end{aligned}
\]
Clearly, the solution of this system is $(0,0,\lambda)$ or $(\mu,0,\mu)$, where $\lambda,\mu\in\Bbbk$ are any constants.
\end{remark}

\subsection{Poisson structure on $\mathrm{Rep}_2(\ac_2)$}
The purpose of this and the next section is to describe the Poisson structure, corresponding to the double Poisson structure $\dbl,\dbr_\alpha$ on the algebra $\ac_2$. Clearly, there are only two one-dimensional representations of the algebra $\ac_2$ (depending on which of the idempotents $e_1,\,e_2$ is sent to $1$ and which is sent to $0$). Hence we begin with real $2$-dimensional representations of $\ac_2$, the space $\mathrm{Rep}_2(\ac_2)$.

\subsubsection{Coordinate ring}
First, consider $e_0$. Since $e_0^2=0$, every $2$-dimensional representation $\rho$ should send it to a $2\times 2$ matrix of rank at most one; any such matrix is equal to the product of two column-vectors $uv^T$, $u^T=(x_1,x_2)$ and $v^T=(y_1,y_2)$. Then the equality $e_0^2=0$ translates as $u\perp v$; without loss of generality we can fix $\|u\|=1$. Also, because $e_2=1-e_1$ we see that in addition to the matrix $\rho(e_0)$ it is enough to fix the matrix $\rho(e_1)$. Since $e_1^2=e_1$, we see that $\rho(e_1)$ is a projector and since $e_1e_0=e_0, e_0e_1=0$, we see that the image of $\rho(e_1)$ is equal to the linear span of $u$, so $\rho(e_1)=uw^T$ for some vector $w=(z_1,z_2)$ such that $(u,w)=1$.



Summing up, we get the following description of $\mathrm{Rep}_2(\ac_2)$ as a real affine variety: every point in the representation space is determined by a $6$-tuple of real parameters, grouped up in three real $2$-vectors, $u=(x_1, x_2),\,v=(y_1,y_2)$ and $w=(z_1,z_2)$, satisfying the conditions
\[
\begin{aligned}
x_1y_1+x_2y_2&=0, & x_1z_1+x_2z_2-1&=0, & x_1^2+x_2^2-1&=0.
\end{aligned}
\]
Clearly, the solutions of this system can be (locally in a generic point) parametrized by three real numbers: $\theta\in[0,2\pi]$ and $\lambda,\mu\in\mathbb R$ so that
\[
\begin{aligned}
u&=(\cos(\theta),\sin(\theta)),\ v=\lambda(-\sin(\theta),\cos(\theta)),\\
 &\,\, w=(\cos(\theta)-\mu\sin(\theta), \sin(\theta)+\mu\cos(\theta)).
\end{aligned}
\]
Below we will describe the Poisson structure on $\mathrm{Rep}_2(\ac_2)$ in terms of these parameters.

\subsubsection{Induced Poisson structure}
We are now going to describe the Poisson structure on $\mathrm{Rep}_2(\ac_2)$ induced by the double Poisson bracket $\dbl,\dbr_\alpha$. Abbreviating $C_{01}^{00}=A$, we have
\[
\{ (e_0)_{ij},(e_1)_{pq} \}
=\dbl e_0 , e_1 \dbr_{pj}' \dbl e_0 , e_1 \dbr_{iq}''
=A(e_0)_{pj}(e_0)_{iq} ,
\]
and all the other brackets of the coordinate functions $(e_a)_{ij},\,a=0,1,2$ are trivial. Recall that in terms of coordinates $(x_1,x_2)$ of vector $u$ and parameters $\lambda,\,\mu$:
\[ 
\rho(e_0)=\begin{pmatrix}
    -\lambda x_1x_2 & \lambda x_1^2 \\ -\lambda x_2^2 & \lambda x_1x_2
\end{pmatrix},\ \rho(e_1)=\begin{pmatrix}
    x_1^2-\mu x_1x_2 & x_1x_2+\mu x_1^2 \\ x_1x_2-\mu x_2^2 & x_2^2+ \mu x_1x_2
\end{pmatrix} 
 \]
 and $(x_1,x_2)=(\cos(\theta),\sin(\theta))$. 
Now from
\[ 
\{ (e_0)_{22},(e_0)_{12} \}=\{ \lambda x_1x_2,\lambda x_1^2 \}=0, 
\]
we deduce
\begin{equation}
\{ \theta,\lambda \}=0.
\end{equation}
Now, to determine $\{\mu,\theta\}$, we recall that for $i,j=1,2$,
\[
\begin{split}
&\{ (e_1)_{12},(e_1)_{ij} \}=\{ x_1x_2+\mu x_1^2,(e_1)_{ij} \}=0,\\
&\{ (e_1)_{21},(e_1)_{ij} \}=\{ x_1x_2-\mu x_2^2,(e_1)_{ij} \}=0,\\
&\{(e_1)_{12}-(e_1)_{21}, (e_1)_{ij}\}=\{ \mu,(e_1)_{ij} \} =0.
\end{split}
\]
This gives the relation
\begin{equation} 
\{ \mu, \theta\}=0.
\end{equation}
Finally, since $\{x_1,x_2\}=\{ \theta,\lambda \}=\{\theta ,\mu \}=0$, we get
\[
\begin{split}
\{(e_0)_{11},(e_1)_{11} \}&=\{-\lambda x_1x_2 , x_1^2-\mu x_1x_2 \}=x_1^2x_2^2\{ \lambda , \mu\}\\
 &=A(e_0)_{11}(e_0)_{11}=A\lambda^2x_1^2x_2^2
 \end{split}
\]
and hence
\begin{equation}
 \{\lambda,\mu \}=A\lambda^2.
\end{equation}

We can also write down the bivector $\pi$, that corresponds to this Poisson structure in matrix form in coordinates $(\theta,\lambda,\mu)$ on representation space:
\[
\pi=\begin{pmatrix}
    0 & 0 & 0 \\ 0 & 0 & A\lambda^2 \\ 0 & -A\lambda^2 & 0
\end{pmatrix}.
\]

\subsection{Poisson structure on $\mathrm{Rep}_3(\ac)$}

In order to get more insight into this structure in higher dimensions, we compute the Poisson structure for dimension three for real representations of $\ac_2$; our computations are similar to the ones we did in dimension two, so we skip all but major steps.

\subsubsection{Coordinate ring}
As in the case of dimension two, the element $e_0$, due to nilpotency, can be mapped only to a matrix of rank no more than one and the image of $e_1$ is a projection matrix $P$ whose image contains the image of $\rho(e_0)$. We need to consider two cases $\mathrm{rk}\,P=1$ and $\mathrm{rk}\,P=2$.

From these and other conditions, we see that a generic $3$-dimensional representation of $\ac_2$ is a point in the following variety.
\begin{itemize}
    \item $\mathrm{rk}\,P=1$.\\ 
    In this case the representations, just like before, are parametrized by three real $3$-vectors: $u=(x_1,x_2,x_3),\,v=(y_1,y_2,y_3),\,w=(z_1,z_2,z_3)$, satisfying the conditions
\[
x_1y_1+x_2y_2+x_3y_3=0,\ x_1z_1+x_2z_2+x_3z_3-1=0,\ x_1^2+x_2^2+x_3^2-1=0.
\]

    \item $\mathrm{rk}\,P=2$.\\
In this case, the equations are the same, although the role of vector-valued parameters is different: rather than parametrizing the points in affine space they deal with the elements in the dual space.
\end{itemize}
Take the case $\mathrm{rk}\,P=1$; then $u$ is a point in the unit sphere $S^2$; then $v$ lies in a plane, orthogonal to $u$ and passing through the origin and $w$ is in the tangent plane of $S^2$ at $u$.

\subsubsection{Induced Poisson structure}
As before, setting $C_{01}^{00}=A$, we have the following commutation relations:
\[
\begin{split}
\{x_iy_j,x_pz_q\}&=\{ (e_0)_{ij},(e_1)_{pq} \}=\lambda x_iy_jx_py_q,\\
\{ x_iy_j,x_py_q \}&=\{ (e_0)_{ij},(e_0)_{pq} \}=0,\\
\{ x_iz_j,x_pz_q \}&=\{ (e_1)_{ij},(e_1)_{pq} \}=0.
\end{split}
\]

As we explained at the end of the previous section the representation space can be described as a triplet of points: one point in the unit sphere a point in the tangent plane and a point in the plane passing through the origin. So we see that there are six free parameters $\theta,\phi,\alpha,\beta,\gamma,\delta$ that describe every representation: ${\theta\in [0,2\pi)}$ and ${\phi\in[-\pi /2,\pi /2}]$ are the spherical coordinates on $S^2$, while $\alpha,\beta,\gamma,\delta$ parametrize the planes. In terms of these coordinates,
\[
\begin{aligned}
(x_1,x_2,x_3)&=(\cos(\theta)\cos(\phi),\sin(\theta)\cos(\phi),\sin(\phi)),\\
(y_1,y_2,y_3)&=(\alpha f_{11} + \beta f_{12},\alpha f_{21} + \beta f_{22},\alpha f_{31} + \beta f_{32}),\\
(z_1,z_2,z_3)&=(x_1+\gamma f_{11} + \delta f_{12},x_2 \gamma f_{21} + \delta f_{22},x_3 + \gamma f_{31} + \delta f_{32})),
\end{aligned}
\]
where $f_{ij}$ are certain explicit functions of $\theta,\phi$ (coordinates of the basis vectors in the planes).

Using these parametrizations we get the following result:
\[\{ \alpha,\gamma \}=\lambda \alpha^2,~\{\alpha,\delta\}=\lambda \alpha\beta,~\{\beta,\gamma\}=\lambda \alpha\beta,~\{\beta,\delta\}=\lambda \beta^2,\]
and all the rest of the coordinate brackets vanish.


We can also represent the Poisson structure in matrix form in coordinates $(\theta,\phi,\alpha,\beta,\gamma,\delta)$ on representation space:
\[
\pi=\begin{pmatrix}
    0 & 0 & 0 & 0 & 0 & 0 \\ 0 & 0 & 0 & 0 & 0 & 0 \\ 
    0 & 0 & 0 & 0 & \lambda \alpha^2 & \lambda \alpha\beta \\ 0 & 0 & 0 & 0 & \lambda \alpha\beta & \lambda \beta^2 \\ 0 & 0 & -\lambda\alpha^2 & -\lambda \alpha\beta & 0 & 0 \\ 0 & 0 & -\lambda\alpha\beta & -\lambda \beta^2 & 0 & 0
\end{pmatrix}.
\]

\section{Discussion}
Given the results for the non-semisimple example, the following step 
is to find an integrable system that can be modeled through this algebra. On the other hand, the simplicity of this
non-semisimple example brings some insight and some questions about double Poisson structures on finite-dimensional non-semisimple algebras, for instance, the
invariance of the rank of the induced Poisson structure.
Further studies may take place with respect to universal enveloping algebras and even the relation of double Poisson brackets with the deformation quantization of an
algebra.

\appendix
\section{Modified double Poisson bracket on $\ac_2$}

In this appendix, we give an example of the modified double Pousson brackets in the sense of Arthamonov (see definition \ref{df:modifieddp}) on the algebra $\ac_2$ of real upper-triangular $2\times 2$-matrices. To this end, we consider the basis $1,e_0,e_1$ in $\ac_2$. We use notation similar to the one we used in the discussion of double Poisson brackets and expand every modified bracket $\dbl e_i,e_j \dbr$ in terms of the corresponding basis in the tensor square of $\ac_2$. As before, the Leibniz rule guarantees that the modified bracket, involving $1$ should be equal to $0$, so we have to consider only the expressions $\dbl e_i,e_j \dbr$ with $i,j=0,1$.



First, we choose the coefficients of $\dbl,\dbr$ so that the Leibniz rules would hold. After that, we verify the $H_0-$skewsymmetry condition, and finally, the Jacobi identity.

\paragraph{Leibniz rules.}
After the calculations, we see that the bracket has to be of the following form:
\begin{equation}\label{mdpb}
\begin{aligned}
    \bullet&\ \dbl e_0,e_0 \dbr=\alpha e_0\otimes e_0 + \beta 1\otimes e_0 + \gamma e_0\otimes 1- (\beta + \gamma )e_0\otimes e_1\\
             &\qquad\qquad\quad  -(\beta + \gamma) e_1 \otimes e_0,\\
    \bullet&\ \dbl e_0,e_1 \dbr=\delta e_0\otimes e_0 + \kappa e_0\otimes 1 + \beta 1\otimes e_1- \kappa e_0\otimes e_1\\
              &\qquad\qquad\quad  -\kappa e_1\otimes e_0 -\beta e_1\otimes e_1,\\
    \bullet&\ \dbl e_1,e_0 \dbr=\iota e_0\otimes e_0 + \gamma e_1\otimes 1 - \gamma e_1\otimes e_1\\,
    \bullet&\ \dbl e_1,e_1 \dbr=\eta e_0\otimes e_0 + \kappa e_1\otimes 1 - \kappa e_1\otimes e_1.    
\end{aligned}
\end{equation}
where $\alpha,\beta,\gamma,\delta,\iota,\kappa,\eta \in \mathbb{K}$ are some constants.

\paragraph{$H_0-$skewsymmetry.} We want  $\{a,b\}+\{b,a\}\in [\ac_2,\ac_2]$ for any $a,b\in \ac$, where $\{-,-\}=m\circ \dbl -,-\dbr$, and $m$ the product of the elements in $\ac_2$.
We have
\begin{itemize}
    \item $\{e_0,e_0\}+\{e_0,e_0\}=2\{e_0,e_0\}=2(\beta e_0+\gamma e_0 - (\beta + \gamma) e_0)=0\in [\ac,\ac]$,
    \item $\{e_1,e_1\}+\{e_1,e_1\}=2\{e_1,e_1\}=2(\kappa e_1 - \kappa e_1)=0\in [\ac,\ac]$,
    \item $\{e_0,e_1\}+\{ e_1,e_0\}= (\kappa e_0 + \beta e_1 -\kappa e_0 - \beta e_1)+(\gamma e_1 - \gamma e_1)=0\in [\ac,\ac]$.
\end{itemize}
So the skew-symmetricity condition holds ``on the nose''.

\paragraph{Jacobi identity.}Since, as we just saw, $\{m,n\}=0$, for any $m,n\in \ac_2$, the Jacobi identity for the bracket $\{,\}$ is automatically satisfied. 

\medskip
So we conclude that all the modified double Poisson brackets in the algebra of upper triangular $2\times 2$-matrices have the form \eqref{mdpb}.

\bibliographystyle{IEEEtran}
\bibliography{references.bib}

@Article{vdb,
    author = {Van den Bergh,Michel},
    title = {Double Poisson Algebras},
    journal = {Transactions of the American Mathematical Society},
    year ={2008} 
}

@Article{sart,
    author = {Arthamonov,Semeon},
    title = {Modified Double Poisson Brackets},
    journal = {J. Algebra},
    pages = {492:212–233},
    year ={2017} 
}

@Article{ors,
    author ={Odesskii, Alexander and Rubtsov,Vladimir and Sokolov, Vladimir } ,
    title = {Double Poisson brackets on free associative algebras},
    journal = {American Mathematical Society},
    year = {2013},
    pages = {225--239}
}

@Article{WL,
    author = {Bao, Wang and Shi-Hao Li},
    title = {On non-commutative leapfrog map},
    year={2023},
    eprint={2310.01993},
    archivePrefix={arXiv},
    primaryClass={math-ph}
}

@book{morita,
    author = {Skowronski, Andrzej and Yamagata, Kunio},
    title = {Frobenius algebras. I: Basic representation theory},
    publisher = {EMS Textbooks in Mathematics},
    year = {2011},
    pages={650}
}

@Article{CBW,
    author = {Crawley-Boevey,W},
    title = {Poisson structures on moduli spaces of representations},
    journal = {J. Algebra},
    pages = { 325:205–215},
    year ={2011} 
}

@book{Loday,
    author = {Loday, J.L.},
    title = {Cyclic Homology},
    publisher = {Springer Berlin, Heidelberg},
    year = {1998},
    pages={516}
}

\end{document}